\documentclass[english]{article}
\usepackage{mathtext}
\usepackage{fontenc}
\usepackage{amssymb}
\usepackage{amsmath}
\usepackage{amsthm}
\usepackage{floatflt}
\usepackage{wrapfig}
\usepackage[T1]{fontenc}
\usepackage[english]{babel}

\usepackage{caption2}

\usepackage{cmap}
\usepackage{hyperref}

\newtheorem{theorem}{Theorem}
\newtheorem{definition}{Definition}
\newtheorem{lemma}{Lemma}
\newtheorem{proposition}{Proposition}
\newtheorem{corollary}{Corollary}

\newtheorem{remark}{Remark}

\begin{document}

\author{\bf Yuriy Tarannikov\footnote{e-mail: yutarann@gmail.com}}

\title{On a Boolean function without bold folding in the spectrum support and implications for greedy approaches to PDT depth}
\date{}

\maketitle

\begin{abstract}
We study Boolean functions and their Fourier spectrum supports in the context of parity decision trees (PDTs). Recently, H.~Hatami et al.~\cite{HHL+} constructed examples whose Fourier support \(\mathcal S\) satisfies
\[
|(\mathcal S+\gamma_1)\cap(\mathcal S+\gamma_2)|=O(|\mathcal S|^{5/6})
\]
for all distinct \(\gamma_1,\gamma_2\), thereby refuting a natural greedy approach based on finding a single large folding direction.
We strengthen this folding estimate by constructing an explicit infinite family of Boolean functions such that
\[
|(\mathcal S+\gamma_1)\cap(\mathcal S+\gamma_2)|=O(|\mathcal S|^{1/2})
\]
for all distinct \(\gamma_1,\gamma_2\). The construction uses a special affine subspace partition, called an APLPS-partition, obtained from full linear spreads. In contrast with the probabilistic construction of \cite{HHL+}, our construction is explicit and has no background spectral components.
We also discuss consequences for greedy approaches to PDT construction. Under the <<lazy>> assumption that the maximum-folding bound is inherited by all restrictions, the usual folding-counting argument cannot yield a PDT upper bound better than \(O(|\mathcal S|^{1/2})\), matching the known general upper bound. However, this inheritance assumption is false in general; hence our result refutes only this <<lazy>> maximum-folding approach, while a complete refutation of adaptive greedy strategies remains open.
\end{abstract}


MSC 68Q11

\section{Introduction}

Boolean functions --- mappings from $\mathbb F_2^n$ to $\{0,1\}$ or $\{-1,1\}$ --- are fundamental objects in complexity theory, cryptography, and discrete mathematics. Their spectral analysis, based on the Fourier transform over the group $\mathbb F_2^n$, allows one to relate combinatorial properties of functions to algebraic characteristics such as matrix rank or spectrum support size. In this work we study the structure of the spectrum support of Boolean functions in the context of parity decision trees (PDTs) and their connection to the log-rank conjecture for XOR functions.

Let us clarify our notation. In cryptography and coding theory (e.g., \cite{Car21,LSSY15,KLT}), Boolean functions are often viewed as mappings to $\{0,1\}$, and the Walsh coefficients are defined as
$$
W_f(u)=\sum_{x\in\mathbb F_2^n} (-1)^{f(x)+\langle u,x\rangle},
$$
where $\langle u,x\rangle = \left(\sum\limits_{i=1}^{n} u_i x_i\right)\bmod 2$.

In complexity theory (in particular, \cite{ODon14,HHL+}), it is common to use functions taking values in $\pm1$ and Fourier coefficients
$$
\widehat f(\alpha)=\frac{1}{2^n}\sum_{x\in\mathbb F_2^n} f(x)(-1)^{\langle \alpha,x\rangle}.
$$
These two approaches are equivalent: if $g(x)=(-1)^{f(x)}$, then $\widehat g(\alpha)=2^{-n}W_f(\alpha)$. The spectrum support (the set of nonzero coefficients) is independent of normalisation. In the sequel we follow the convention of \cite{HHL+}, i.e. we consider functions with values in $\pm1$ and use the normalised Fourier transform unless stated otherwise.

For a Boolean function $f:\mathbb F_2^n\to\{-1,1\}$, we define its spectrum support
$$
\mathcal S=\operatorname{supp}\widehat f=\{\alpha\in\mathbb F_2^n\mid \widehat f(\alpha)\ne0\},
$$
and denote $k=|\mathcal S|$.

A central notion in our study is that of a \emph{folding} in the spectrum support. Following \cite{MS24}, for $\gamma\in\mathbb F_2^n$ we define the set
$$
\mathcal O_\gamma=\left\{(\alpha,\beta)\in\binom{\mathcal S}{2}\mid \alpha+\beta=\gamma\right\}.
$$
The elements of $\mathcal O_\gamma$ are unordered pairs of distinct vectors from $\mathcal S$ whose sum is $\gamma$. Such pairs are said to form a \emph{folding} in direction $\gamma$. If $|\mathcal O_\gamma|\ge 2$, then the direction $\gamma$ contains at least two distinct pairs in the folding; if $|\mathcal O_\gamma|$ is large, we speak of a <<bold>> folding.

In \cite{HHL+} an equivalent but less transparent notation is used, namely intersections of shifts: for distinct $\gamma_1,\gamma_2$ one considers
$$
(\mathcal S+\gamma_1)\cap(\mathcal S+\gamma_2),
$$
where $\mathcal S+\gamma=\{\alpha+\gamma\mid \alpha\in\mathcal S\}$. It is easy to see that if $\gamma=\gamma_1+\gamma_2$, then
$$
|(\mathcal S+\gamma_1)\cap(\mathcal S+\gamma_2)| = 2|\mathcal O_\gamma|,
$$
since each intersection gives a pair of endpoints corresponding to one pair in $\mathcal O_\gamma$. Thus the two descriptions are equivalent up to a factor of~2.

The study of the structure of folding relies on the classical theorem of Titsworth \cite{Tit63}, which gives a necessary and sufficient condition for a set of coefficients $\{\widehat f(\alpha)\}$ to correspond to a Boolean function. For our purposes, the following consequence is important: for any nonzero $\gamma$,
$$
\sum_{\alpha\in\mathbb F_2^n} \widehat f(\alpha)\widehat f(\alpha+\gamma)=0.
$$
In particular, this forbids the existence of isolated pairs in $\mathcal O_\gamma$: if $(\alpha,\beta)\in\mathcal O_\gamma$, then there must exist at least one other pair from the spectrum support in the same direction so that the corresponding products cancel. This property is crucial when analysing possible folding sizes.

We now turn to the motivation coming from communication complexity. For a two-party Boolean function $F:X\times Y\to\{-1,1\}$, its communication complexity $\mathrm{CC}(F)$ is the minimum number of bits that must be exchanged between two players to compute $F(x,y)$ given their respective inputs. The log-rank conjecture of Lovasz and Saks \cite{LS88} states that $\mathrm{CC}(F)\le \operatorname{polylog}(\operatorname{rank} M_F)$, where $M_F$ is the communication matrix of size $|X|\times|Y|$ and the rank is over $\mathbb R$. Despite much effort, this conjecture remains open.

Of particular interest is the class of XOR functions, where $F(x,y)=f(x\oplus y)$ for some $f:\mathbb F_2^n\to\{-1,1\}$. In this case the rank of the matrix of $F$ is exactly $|\mathcal S|$, the size of the spectrum support of $f$ (proved in \cite{BC99}, see also \cite{Gib24}). Moreover, the work \cite{HHL18} establishes a polynomial equivalence between $\mathrm{CC}(f\circ\oplus)$ and the parity decision tree complexity (PDT) of $f$. A PDT is a binary tree where each internal node is labelled by a character $\chi_\alpha$ (i.e. it queries the parity $\langle\alpha,x\rangle$ of a subset of bits), and leaves are labelled with $\pm1$. The depth of a PDT for $f$, denoted $\mathrm{PDT}(f)$, provides a natural upper bound on $\mathrm{CC}(f\circ\oplus)$. A PDT of depth $D$ gives a deterministic protocol of cost $O(D)$, since each parity query of $x\oplus y$ can be answered with constant communication.

Currently, it is known that for every Boolean function $f$, $\mathrm{PDT}(f)=O(\sqrt{k})$, where $k=|\mathcal S|$ \cite{TWXZ13,MS24}. This corresponds to the upper bound $O(\sqrt{\operatorname{rank}})$ for communication complexity, which was recently obtained also in the general case \cite{ST25}.

To make progress towards a polylogarithmic bound, a so-called \emph{greedy approach} has been proposed: at each step, find a direction $\gamma$ with a bold folding size $|\mathcal O_\gamma|$ (equivalently, a large intersection $|(\mathcal S+\gamma_1)\cap(\mathcal S+\gamma_2)|$), query the corresponding parity $\langle\gamma,x\rangle$, and recurse on restrictions whose spectrum support shrinks. If at each step one can find a folding of size at least $k/\operatorname{polylog}(k)$, then the PDT depth would be polylogarithmic.

However, a recent work \cite{HHL+} showed that there exist Boolean functions for which all folding have size at most $O(k^{5/6})$. This rules out the possibility of obtaining polylogarithmic depth via the <<lazy>> greedy approach that simply substitutes a universal bound at each step. Moreover, in the terminology of \cite{HHL+}, this means that the corresponding greedy method cannot yield an upper bound better than $\widetilde O(k^{1/6})$, and with an optimized choice of parameters one can rule out bounds better than $\widetilde O(k^{1/5})$.

In the present work we significantly strengthen the result of \cite{HHL+}. We construct an explicit infinite family of Boolean functions for which the maximum folding size in the spectrum support is at most $O(k^{1/2})$. This is achieved by means of a special affine subspace partition (APLPS) based on full linear spreads. Our construction avoids the use of <<background>> components present in \cite{HHL+} and, in contrast to their probabilistic construction, is fully explicit.

The main result can be stated informally as follows.

\begin{theorem}[Informal version]
For every $d\ge1$ there exists a Boolean function $f$ on $n=(2d+1)+2^{d+1}$ variables with spectrum support $\mathcal S$, $|\mathcal S|=2^{2d+2}$, such that for any distinct $\gamma_1,\gamma_2$,
$$
|(\mathcal S+\gamma_1)\cap(\mathcal S+\gamma_2)| = O(|\mathcal S|^{1/2}).
$$
\end{theorem}

This is a substantial improvement in the exponent in the \cite{HHL+} bound, from $5/6$ to $1/2$.
Correspondingly, under the <<lazy>> greedy framework, i.e. assuming that the same maximum-folding estimate is inherited by all restrictions, the standard analysis would force a branch of length $\Omega(k^{1/2})$. Since the general upper bound $O(k^{1/2})$ is already known, such an analysis cannot improve the state of the art.

We also discuss in detail that the inheritance assumption is not justified in general, and we provide examples (the classical address function) where it is clearly violated. Thus, it is precisely the <<lazy>> variant of the greedy approach that is refuted; the possibility of constructing an adaptive greedy algorithm that analyses the structure of restrictions remains open.

The paper is organised as follows. In Section~2 we discuss generalisations of the address function from \cite{KLT} and \cite{HHL+}. In Section~3 we use APLPS-partitions to prove the main strengthening theorem. In Section~4 we discuss the consequences for greedy approaches to PDT depth and their limitations.

\section{Generalizations of the address function}

The address function is a classical object in Boolean function theory, complexity theory, and cryptography. For a positive integer $m$, it is defined on $n=m+2^m$ variables:
$$
\operatorname{Add}_m(x_1,\dots,x_m,y_0,\dots,y_{2^m-1}) = (-1)^{y_{\xi(x_1,\dots,x_m)}},
$$
where $\xi(x)=\sum\limits_{i=1}^m x_i 2^{m-i}$ interprets the binary vector as an integer. This function has small decision tree depth ($m+1$), but its Fourier/Walsh--Hadamard spectrum has size $|\operatorname{supp}\widehat{\operatorname{Add}_m}|=2^{2m}$, which makes it an important example in the study of the log-rank conjecture for XOR functions. Moreover, the address function is often used in coding theory and in constructions of plateaued, correlation-immune, and other cryptographically important functions.

\subsection{Generalization in \cite{KLT}}

In \cite{KLT}, the following construction of a spectrum support is proposed. Let $m_1,m_2$ be positive integers, and for each $i=1,\dots,m_2$ let a linear subspace $L_i\subseteq \mathbb F_2^{m_1}$ and a shift vector $a^i\in\mathbb F_2^{m_1}$ be given. Consider the set
$$
S=\{(v,u)\in \mathbb F_2^{m_1+m_2} : u=e_i,\; v\in L_i+a^i,\; i=1,\dots,m_2\},
$$
where $e_i$ is the $i$-th standard basis vector in $\mathbb F_2^{m_2}$. This set is declared to be the \emph{spectrum support} of a hypothetical Boolean function on $m_1+m_2$ variables; the existence of such a function is not guaranteed.

In \cite{KLT}, the Walsh--Hadamard coefficients are denoted by $W_f(u)$, and the Fourier coefficients of auxiliary functions by $F_{\widehat f_i}(x)$. For a given $S$, the functions $\widehat f_i(v)=W_f(v,e_i)$ on $\mathbb F_2^{m_1}$ are defined, and their Fourier transform is
$$
F_{\widehat f_i}(x)=\sum_{v\in\mathbb F_2^{m_1}} \widehat f_i(v)(-1)^{\langle x,v\rangle}.
$$
Then the following characterization holds.

\begin{proposition}[{\cite[Proposition~5]{KLT}}]\label{prop:KLT}
The set of coefficients $\{F_{\widehat f_i}(x)\}$ defines a Boolean function on $\mathbb F_2^{m_1+m_2}$ with spectrum support contained in
$S^*=\{(v,e_i)\,:\, i=1,\dots,m_2,\ v\in \mathbb{F}_2^{m_1}\}$
if and only if for each $x\in\mathbb F_2^{m_1}$ exactly one of the coefficients $F_{\widehat f_i}(x)$ equals $\pm 2^{m_1+m_2}$, and all the others are zero.
\end{proposition}

In particular, when the natural counting condition
$$
\sum\limits_{i=1}^{m_2} 2^{m_1-\dim L_i}=2^{m_1}
$$
holds, one obtains the following consequence.

\begin{corollary}[{\cite[Corollary~4]{KLT}}]\label{cor:KLT}
If a Boolean function with spectrum support $S$ exists and
\begin{equation}\label{eq:tag_1}
\sum_{i=1}^{m_2} 2^{m_1-\dim L_i}=2^{m_1},
\end{equation}
then for each $i$, the number of nonzero coefficients $F_{\widehat f_i}(x)$ is exactly $2^{\dim L_i^\perp}=2^{m_1-\dim L_i}$, where
$$
L_i^\perp = \{ u \in \mathbb{F}_2^{m_1} \mid \langle u, v\rangle = 0 \text{ for all } v\in L_i\},\quad
\langle u,v\rangle = \left(\sum_{j=1}^{m_1} u_j v_j\right)\bmod 2.
$$
\end{corollary}

Equation~\eqref{eq:tag_1} is necessary for the total number of elements in the orthogonal (dual) subspaces $L_i^\perp$ to be exactly $2^{m_1}$.

Condition \eqref{eq:tag_1} naturally leads to the problem of partitioning the space $\mathbb F_2^{m_1}$ into affine subspaces. Indeed, if for each $i$ the number of nonzero coefficients $F_{\widehat f_i}(x)$ is exactly $2^{m_1-\dim L_i}$, then these coefficients must be concentrated on some shift of the subspace $L_i^\perp$.

Thus, in this situation, the nonzero values of the functions $F_{\widehat f_i}$ naturally form a partition of $\mathbb F_2^{m_1}$ by affine subspaces of the form $L_i^\perp+b$.
In \cite{KLT}, the notation
$$
N(r_1,L_{i_1}^{\perp};\dots;r_s,L_{i_s}^{\perp})
$$
is introduced for the number of ways to partition $\mathbb F_2^{m_1}$ into affine subspaces among which exactly $r_j$ have the form $L_{i_j}^{\perp}+b$. (The original text uses the term <<linear subspaces>>, but from the form $L_{i_j}^{\perp}+b$ it is clear that affine subspaces are meant.)

We note that \cite{KLT} has also been cited in the complexity-theoretic literature; for instance, it appears in \cite{Sanyal19}, which studies Boolean functions with sparse Fourier support in connection with the log-rank conjecture.

\subsection{Generalization in \cite{HHL+}}

In \cite{HHL+} (Example 3.3), the following construction is considered. Let $k\ge 3$, $m_1=7k$, $m_2=2^k$. Let $\{A_i\}_{i=1}^{m_2}$ be a family of pairwise disjoint affine subspaces in $\mathbb F_2^{7k}$ of the form $A_i=V_i+a_i$, where $V_i$ are linear subspaces in $\mathbb F_2^{7k}$ with $\dim V_i=2k$, and $a_i\in\mathbb F_2^{7k}$ are shift vectors. Define the function
$$
f(x,y)=1+\sum_{i=1}^{2^k} \mathbf 1_{A_i}(x)\bigl((-1)^{y_i}-1\bigr),\qquad x\in\mathbb F_2^{7k},\; y\in\mathbb F_2^{2^k}.
$$
This definition is equivalent to the following: $f(x,y)=(-1)^{y_i}$ if $x\in A_i$ for some $i$, and $f(x,y)=1$ otherwise (i.e., for $x\notin \bigcup\limits_i A_i$). This representation is convenient for computing the spectrum.

The numerical choice ($d=2k$, $m=7k$) is tailored to the estimates in \cite{HHL+}. Similar constructions can be considered with other parameters $(d,m)$, provided suitable disjointness and intersection conditions are satisfied.

Using the standard formula for the Fourier transform of the indicator of an affine subspace (see \cite[Fact~2.1]{HHL+}), one obtains that the spectrum support $\mathcal S=\operatorname{supp}\widehat f$ satisfies the inclusions
\begin{equation}\label{eq:tag_2_17}
\bigcup_{i=1}^{2^k} \bigl(V_i^\perp\times\{e_i\}\bigr)
\subseteq
\mathcal S
\subseteq
\left(
\bigcup_{i=1}^{2^k} \bigl(V_i^\perp\times\{0\}\bigr)
\right)
\cup
\left(
\bigcup_{i=1}^{2^k} \bigl(V_i^\perp\times\{e_i\}\bigr)
\right).
\end{equation}
where $V_i^\perp$ is the orthogonal (or dual) subspace to $V_i$ in $\mathbb F_2^{7k}$ (of dimension $5k$), and $e_i$ are the standard basis vectors in $\mathbb F_2^{2^k}$.

In the language of the generalization from \cite{KLT}, this can be rewritten in a <<similar style>> by separating the right-hand part $u$:
$$
\Bigl\{(v,e_i): i=1,\dots,2^k,\; v\in V_i^\perp\Bigr\}
\subseteq\mathcal S\subseteq
\Bigl\{(v,0): v\in \bigcup_{i=1}^{2^k} V_i^\perp\Bigr\}
\;\cup\;
\Bigl\{(v,e_i): i=1,\dots,2^k,\; v\in V_i^\perp\Bigr\}.
$$
Here, the possible $u=0$ part is contained in a union of many distinct subspaces $V_i^\perp$, whereas in the definition from \cite{KLT}, each $u$ corresponds to exactly one affine subspace $L_u+a_u$. Therefore, the construction of \cite{HHL+} is \emph{not} a special case of the generalization from \cite{KLT}. Note also that in \cite{HHL+}, the linear subspaces $V_i$ and their dual subspaces play the opposite role to that in \cite{KLT}: in the latter, $L_i$ are subspaces in the left part, and the support contains $L_i^\perp$; in \cite{HHL+}, the support directly contains $V_i^\perp$, while $V_i$ are the directions of the original affine sets.

The possible appearance of the $u=0$ layer in (\ref{eq:tag_2_17}) is caused by the <<background>> term
$$
1-\sum\limits_{i=1}^{2^k} \mathbf 1_{A_i}(x),
$$
which is the indicator of the complement of $\bigcup\limits_{i=1}^{2^k} A_i$. This term is necessary in the construction in \cite{HHL+} because the affine subspaces $A_i$ do not form a partition of $\mathbb F_2^{7k}$; without it, the expression would vanish outside the union and hence would not define a $\{\pm1\}$-valued Boolean function.

Thus, unlike the \cite{KLT}-type situation where the relevant affine pieces form a partition, the family in \cite{HHL+} requires an additional <<background>> term to make the function Boolean-valued on the whole space.

Another generalization of the address function for similar purposes is introduced in \cite[proof of Theorem~3.8, part 1]{MS24}. There, too, a Boolean function is constructed with certain (but different) properties of folding in the spectrum support.

\subsection{Connection to the subsequent exposition}

In the next section, we will propose a construction based on an \emph{APLPS-partition} of the space $\mathbb F_2^m$ into affine subspaces, which avoids background components. Thereby we obtain a stronger bound on the maximum intersection of spectrum supports ($O(|\mathcal S|^{1/2})$ instead of $O(|\mathcal S|^{5/6})$), which gives a better refutation of the greedy approach.

Our construction will be
\begin{itemize}
\item
a special case of the \cite{KLT}-type support construction: for each $u=e_i$, the corresponding left part is a single linear subspace $V_i^\perp$, and there is no additional $u=0$ layer;
\item
a variant with a stronger folding estimate than in the construction in \cite{HHL+}.
\end{itemize}

\section{Strengthening the result of H.~Hatami et al. from \cite{HHL+}}

In this section, we improve the estimate of Theorem 1.3 from \cite{HHL+}, replacing the exponent $5/6$ by $1/2$. To this end, we use a special partition of the space into affine subspaces (APLPS-partition), which allows us to construct a function without <<background>> spectral components and at the same time control the size of the largest (boldest) folding in the spectrum support.

We recall the main results of \cite{HHL+} that we will strengthen.

\begin{lemma}\label{lemma_from_HHL+}\cite[Lemma~3.4]{HHL+}
For $k\ge 3$ there exists a family of affine subspaces $A_i=V_i+a_i\subset \mathbb F_2^{7k}$, where $V_i$ are linear subspaces in $\mathbb F_2^{7k}$, $a_i\in\mathbb F_2^{7k}$ are shift vectors, $i=1,\dots,2^k$, such that:
\begin{enumerate}
\item $\dim V_i=2k$;
\item $A_i\cap A_j=\varnothing$ for $i\ne j$;
\item $V_i\cap V_j=\{0\}$ for $i\ne j$;
\item for every nonzero $v\in\mathbb F_2^{7k}$, $|\{i: v\in V_i^\perp\}|\le 7$.
\end{enumerate}
\end{lemma}

Based on this family, in \cite{HHL+} the function
$$
f(x,y)=1+\sum_{i=1}^{2^k} \mathbf 1_{A_i}(x)\bigl((-1)^{y_i}-1\bigr),\qquad x\in\mathbb F_2^{7k},\; y\in\mathbb F_2^{2^k}
$$
is defined, and it is proved that its spectrum support $\mathcal S$ has size at least $2^{6k}$, and for any distinct $\gamma_1,\gamma_2$,
$$
|(\mathcal S+\gamma_1)\cap(\mathcal S+\gamma_2)| \le 2^{5k+4}=O(|\mathcal S|^{5/6}).
$$
This leads to Theorem 1.3:

\begin{theorem}\cite[Theorem~1.3]{HHL+}
For infinitely many $n$ there exists a Boolean function $f:\mathbb F_2^n\to\{-1,1\}$ such that for $\mathcal S=\operatorname{supp}\widehat f$ and any distinct $\gamma_1,\gamma_2$,
$$
|(\mathcal S+\gamma_1)\cap(\mathcal S+\gamma_2)| = O(|\mathcal S|^{5/6}).
$$
\end{theorem}

We will show that one can construct a function with the estimate $O(|\mathcal S|^{1/2})$, which is significantly stronger.

Note also that the existence of a system of disjoint subspaces with the parameters specified in Lemma~\ref{lemma_from_HHL+} was established in \cite{HHL+} by a probabilistic, i.e. completely nonconstructive method, whereas all our constructions will be explicit.

For this, we first introduce the notion of APLPS-partition.

\begin{definition}
An $(m,d)$-APLPS-partition (\textit{affine partition, linear partial spread}) of the space $\mathbb F_q^m$ is a family $\{A_i\}_{i\in I}$ of affine subspaces of dimension $d$ satisfying:
\begin{itemize}
\item $A_i=V_i+a_i$, where $V_i$ is a linear subspace of dimension $d$, $a_i\in\mathbb F_q^m$;
\item the $A_i$ are pairwise disjoint and cover $\mathbb F_q^m$ (a partition);
\item $V_i\cap V_j=\{0\}$ for $i\ne j$.
\end{itemize}
\end{definition}

We shall need the following facts.

Lemma~\ref{spread_lemma_1} is well known. Its proof for $q=2$ is contained in \cite{Dil68}. We give Lemma~\ref{spread_lemma_1} with proof not only for completeness, but also because the construction from the proof of Lemma~\ref{spread_lemma_1} will be used in the proof of Lemma~\ref{spread_lemma_3}.

\begin{lemma}\label{spread_lemma_1}{\bf (full linear spread)}
For any even $N$ and any finite field $\mathbb{F}_q$, there exists a full linear spread in $\mathbb{F}_q^N$, i.e. a family of $q^{N/2}+1$ linear subspaces of dimension $N/2$, pairwise intersecting in $\{0\}$, such that every nonzero vector belongs to exactly one of them.
\end{lemma}

\begin{proof}
Put $k = N/2$. Identify $\mathbb{F}_q^k$ with the finite field $\mathbb{F}_{q^k}$ using a fixed basis. Represent $\mathbb{F}_q^N$ as the set of pairs $(x, y)$ with $x, y \in \mathbb{F}_{q^k}$. For each $\alpha \in \mathbb{F}_{q^k}$, define the linear subspaces:
$$
V_\alpha = \{(x, \alpha x) \mid x \in \mathbb{F}_{q^k}\},
$$
and add
$$
V_\infty = \{(0, y) \mid y \in \mathbb{F}_{q^k}\}.
$$
Each of these sets is an $\mathbb{F}_q$-linear subspace of dimension $k$ (componentwise addition). If $(x, y) \in V_\alpha \cap V_\beta$ for $\alpha \neq \beta$, then $x = 0$ and hence $y = 0$; similarly $V_\alpha \cap V_\infty = \{0\}$. Thus all subspaces are pairwise intersecting only at zero.

Any nonzero vector $(x, y) \in \mathbb{F}_q^N$ belongs to $V_\alpha$ with $\alpha = y x^{-1}$ (if $x \neq 0$), and belongs to $V_\infty$ if $x = 0$. Therefore the union of all subspaces covers all nonzero vectors, and we obtain a full linear spread.
\end{proof}

Full linear spreads are well known in cryptography, since the $\mathcal{PS}$ family ({\it partial spread} family) of bent functions is based on the partial spreads extracted from them \cite{Dil68,Car21,Mes16}.

The proof of Lemma~\ref{spread_lemma_2} uses the full spread from Lemma~\ref{spread_lemma_1} and a hyperplane section. A similar construction (corresponding to the case $m=2d+1$) is given in \cite[Proposition 7]{BFIK25}.

\begin{lemma}\label{spread_lemma_2}{\bf (existence of APLPS-partition)}
Let $m \ge 2d + 1$. Then there exists a partition of $\mathbb{F}_q^m$ into affine subspaces of dimension $d$ with linear subspaces pairwise intersecting in $\{0\}$ (an $(m,d)$-APLPS-partition).
\end{lemma}

\begin{proof}
The case $d = 0$ is trivial. Let $d \ge 1$. Put $k = m - d$, then $N = 2k = 2(m - d)$ is even. By Lemma~\ref{spread_lemma_1}, there exists a full linear spread in $\mathbb{F}_q^N$:
$$
\mathcal{S} = \{V_\alpha\}_{\alpha \in I} \cup \{V_\infty\},
$$
where $|I| = q^k$, $\dim V_\alpha = \dim V_\infty = k$, and any two distinct subspaces intersect in $\{0\}$.

Since the scalar product is nondegenerate, $\dim V_\infty^\perp = N - k = k = m-d$.

From $m \ge 2d+1$ we get $m-2d \ge 1$ and $m-2d \le k$. Choose linearly independent vectors
$$
u_1, \dots, u_{m-2d} \in V_\infty^\perp.
$$

For each $\alpha \in I$, the space $\mathbb{F}_q^N$ is the direct sum $V_\infty \oplus V_\alpha$. The map $u \mapsto \langle u, \cdot \rangle|_{V_\alpha}$ gives an isomorphism $V_\infty^\perp \to V_\alpha^*$ (injectivity: if $\langle u, x\rangle = 0$ for all $x \in V_\alpha$, then $u$ is orthogonal to the whole $V_\infty \oplus V_\alpha = \mathbb{F}_q^N$, so $u=0$; the dimensions coincide). Hence the linear forms $\langle u_1, \cdot \rangle, \dots, \langle u_{m-2d}, \cdot \rangle$, restricted to $V_\alpha$, are linearly independent.

Define the affine subspace
$$
A = \bigcap_{j=1}^{m-2d} \{ x \in \mathbb{F}_q^N \mid \langle u_j, x\rangle = 1 \}.
$$
Since the $u_j$ are independent, this system is consistent and $\dim A = N - (m-2d) = m$.

For $\alpha \in I$, put $\widetilde{A}_\alpha = V_\alpha \cap A$. From the independence of the restrictions of $\langle u_j, \cdot \rangle$ to $V_\alpha$, it follows that the system of equations $\langle u_j, x\rangle = 1$ ($j=1,\dots,m-2d$) on $V_\alpha$ is consistent; its solution set is an affine subspace of dimension
$$
\dim V_\alpha - (m-2d) = k - (m-2d) = d.
$$

If $\alpha \neq \beta$, then $\widetilde{A}_\alpha \cap \widetilde{A}_\beta \subseteq V_\alpha \cap V_\beta = \{0\}$, but $0 \notin A$ (since $\langle u_j, 0\rangle = 0 \neq 1$), so the $\widetilde{A}_\alpha$ are pairwise disjoint.

The cardinality of the union is
$$
\Bigl| \bigcup_{\alpha \in I} \widetilde{A}_\alpha \Bigr| = |I| \cdot q^d = q^k \cdot q^d = q^m = |A|,
$$
hence $\{\widetilde{A}_\alpha\}_{\alpha \in I}$ is a partition of $A$.

The linear subspace whose shift is $\widetilde{A}_\alpha$ is
$$
\widetilde{V}_\alpha = V_\alpha \cap \bigcap_{j=1}^{m-2d} \{ x \mid \langle u_j, x\rangle = 0 \},
$$
and $\widetilde{V}_\alpha \subset V_\alpha$, so $\widetilde{V}_\alpha \cap \widetilde{V}_\beta = \{0\}$ for $\alpha \neq \beta$.

It remains to choose an arbitrary affine isomorphism $\varphi : A \to \mathbb{F}_q^m$ (a shift and a linear bijection). The images $\varphi(\widetilde{A}_\alpha)$ form the desired APLPS-partition of $\mathbb{F}_q^m$.
\end{proof}

\begin{remark}{\bf (necessary condition)}
If $d < m \le 2d$ (i.e. $m \ge 2d+1$ is violated and the partition is nondegenerate), an
$(m,d)$-APLPS-partition does not exist. Indeed, let $\{A_i = V_i + a_i\}$ be such a partition.
If $m<2d$, then by Grassmann's formula $\dim(V_i\cap V_j)\ge 2d-m>0$, contradicting $V_i\cap V_j=\{0\}$.
It remains to consider $m=2d$. After translating the whole partition, assume $A_1=V_1$ and $A_2=V_2+a$. Since $V_1\cap V_2=\{0\}$ and $\dim V_1+\dim V_2=m$, we have
$\mathbb F_q^m=V_1\oplus V_2$.
Write $a=v_1+v_2$, with $v_i\in V_i$. Then
$v_1=a-v_2\in V_2+a=A_2$, and also $v_1\in V_1=A_1$, a contradiction.
\end{remark}

For $q=2$ and $m=2d+1$ (the more general case will not be needed), we can achieve an additional property. The proof of Lemma~\ref{spread_lemma_3} is based on the standard model of the full spread via the field $\mathbb F_{2^{d+1}}$ and the choice of an element with trace equal to one.

\begin{lemma}\label{spread_lemma_3}{\bf (orthogonal subspace property for $q=2$, $m=2d+1$)}
Let $q = 2$, $m = 2d+1$. Then there exists an $(m,d)$-APLPS-partition such that for any nonzero $z \in \mathbb{F}_2^m$, the number of linear subspaces $V_i$ of the partition for which $z \in V_i^\perp$ is at most two.
\end{lemma}

\begin{proof}
For $m = 2d+1$, the parameters in Lemma~\ref{spread_lemma_2} are:
$k = m - d = d+1$, $N = 2k = 2d+2$, and $m-2d = 1$. Thus one vector $u \in V_\infty^\perp \setminus \{0\}$ is used. The partition is constructed from a full linear spread in $\mathbb{F}_2^{2k}$; the linear subspaces of the partition (after the isomorphism) are
$$
\widetilde{V}_\alpha = V_\alpha \cap u^\perp \quad (\alpha \in I),
$$
where $u^\perp = \{ x \in \mathbb{F}_2^{2k} \mid \langle u, x\rangle = 0\}$.

Realise the full spread in the standard way via the field $\mathbb{F}_{2^k}$ as in the proof of Lemma~\ref{spread_lemma_1}. Identify $\mathbb{F}_2^k$ with $\mathbb{F}_{2^k}$ and represent $\mathbb{F}_2^{2k}$ as the set of pairs $(x, y)$ with $x, y \in \mathbb{F}_{2^k}$. Define the scalar product by
$$
\langle (u_1,u_2), (x,y) \rangle = \operatorname{Tr}(u_1 x + u_2 y),
$$
where $\operatorname{Tr}: \mathbb{F}_{2^k} \to \mathbb{F}_2$ is the absolute trace. This is a nondegenerate bilinear form.
Set
$$
V_\alpha = \{(x, \alpha x) \mid x \in \mathbb{F}_{2^k}\}, \qquad
V_\infty = \{(0, y) \mid y \in \mathbb{F}_{2^k}\}.
$$
Then $V_\infty^\perp = \{ (a,0) \mid a \in \mathbb{F}_{2^k} \}$, since $\langle (a,b), (0,y)\rangle = \operatorname{Tr}(b y) = 0$ for all $y$ implies $b=0$. Choose $u = (a,0) \in V_\infty^\perp$ with $\operatorname{Tr}(a) = 1$ (such $a$ exists since the trace is surjective). Then
$$
u^\perp = \{ (x,y) \mid \operatorname{Tr}(a x) = 0 \},
$$
and, according to the proof of Lemma~\ref{spread_lemma_2}, the linear directions of the affine subspaces of the constructed $(m,d)$-APLPS-partition (up to an affine isomorphism $A = \{ \langle u, x\rangle = 1 \} \cong \mathbb{F}_2^m$) are
$$
\widetilde{V}_\alpha = V_\alpha \cap u^\perp = \{ (x, \alpha x) \mid \operatorname{Tr}(a x) = 0 \}.
$$

Let $z \in \mathbb{F}_2^m$ be a nonzero vector.

The affine isomorphism \(A\to\mathbb F_2^m\) has a linear part \(T:u^\perp\to\mathbb F_2^m\).
For the images \(T\widetilde V_\alpha\), the condition
\[
z\in (T\widetilde V_\alpha)^\perp
\]
is equivalent to
\[
T^t z\in \widetilde V_\alpha^\perp.
\]
Since \(T^t\) is a bijection, it suffices to prove the required multiplicity bound for the original
spaces \(\widetilde V_\alpha\subset u^\perp\).
We may assume $z \in u^\perp$ and write $z = (u_0, v_0)$ with $\operatorname{Tr}(a u_0) = 0$.

We determine for which $\alpha$ we have $z \perp \widetilde{V}_\alpha$, i.e. $\widetilde{V}_\alpha \subseteq z^\perp$. This is equivalent to: for every $x \in \mathbb{F}_{2^k}$ with $\operatorname{Tr}(a x)=0$,
$$
0 = \langle (u_0, v_0), (x, \alpha x) \rangle = \operatorname{Tr}(u_0 x + v_0 \alpha x) = \operatorname{Tr}\bigl((u_0 + \alpha v_0) x\bigr).
$$
Thus the linear map $x \mapsto \operatorname{Tr}(w x)$ with $w = u_0 + \alpha v_0$ is identically zero on the hyperplane $\{x \mid \operatorname{Tr}(a x) = 0\}$. This is possible only if $w$ is proportional to $a$ with coefficient in $\mathbb{F}_2$. Hence
$$
u_0 + \alpha v_0 = c a, \qquad c \in \{0,1\}.
$$

Consider two cases.
\begin{enumerate}
\item $v_0 \neq 0$. Then from the equation we find $\alpha = u_0 v_0^{-1}$ (if $c=0$) or $\alpha = (u_0 + a) v_0^{-1}$ (if $c=1$). We get at most two solutions.
\item $v_0 = 0$. Then $u_0 = c a$. Since $z \in u^\perp$, we have $\operatorname{Tr}(a u_0) = 0$. If $c=1$, then $u_0 = a$ and $\operatorname{Tr}(a u_0) = \operatorname{Tr}(a^2) = \operatorname{Tr}(a) = 1$ by the choice of $a$, a contradiction. Hence $c=0$, so $u_0 = 0$, giving $z = (0,0)$, contradicting the choice of nonzero $z$. Thus the case $v_0 = 0$ is impossible for nonzero $z$.
\end{enumerate}

Therefore, for any nonzero $z$, there are at most two indices $\alpha$ with $z \perp \widetilde{V}_\alpha$. Under the affine isomorphism, this property means that in the constructed
APLPS-partition of $\mathbb{F}_2^m$, for every nonzero vector, its dual subspace contains at most two linear subspaces $V_i$ from the partition.
\end{proof}

Thus for $m = 2d+1$ and $q = 2$, the constructed $(m,d)$-APLPS-partition has the additional strong property: every nonzero dual subspace to a direction contains at most two directions of the partition.

Now we construct the function. Let $d\ge1$, $m=2d+1$. Take an $(m,d)$-APLPS-partition $\{A_i\}_{i=1}^N$, where $N=2^{m-d}=2^{d+1}$, $A_i=V_i+a_i$, $\dim V_i=d$, satisfying Lemma~\ref{spread_lemma_3}. Define
$$
f(x,y)=\sum_{i=1}^{N} \mathbf 1_{A_i}(x)(-1)^{y_i},\qquad x\in\mathbb F_2^m,\; y\in\mathbb F_2^N.
$$
This is a well-defined Boolean function since the sets $A_i$ partition $\mathbb F_2^m$.

We compute its spectrum support. Unlike \cite{HHL+}, there is no constant term, so by analogy with the computations in \cite{HHL+} (see the proof of Claim 3.5) we obtain\footnote{It is precisely through this spectrum support $\mathcal S$ that the generalization of the address function in \cite{KLT} was defined.}
$$
\mathcal S=\operatorname{supp}\widehat f = \bigcup_{i=1}^N \bigl(V_i^\perp \times \{e_i\}\bigr),
$$
where $e_i$ is the $i$-th unit vector in $\mathbb F_2^N$. The size of the support is
$$
|\mathcal S| = N\cdot 2^{\dim V_i^\perp} = 2^{d+1}\cdot 2^{m-d} = 2^{d+1}\cdot 2^{d+1} = 2^{2d+2}.
$$
Thus $|\mathcal S|^{1/2} = 2^{d+1}$.

We estimate the maximum intersection $|\mathcal S\cap(\mathcal S+\gamma)|$ for an arbitrary nonzero $\gamma=(\gamma_x,\gamma_y)\in\mathbb F_2^{m+N}$. Consider the possible values of the weight of $\gamma_y$.

\begin{itemize}
\item If $|\gamma_y|>2$ or $|\gamma_y|=1$, the intersection is empty, since all vectors in $\mathcal S$ have $y$-part equal to one of the $e_i$, and the difference of two such vectors can have weight $0$ or $2$.

\item If $|\gamma_y|=2$, then $\gamma_y=e_i+e_j$ for some $i\ne j$. In this case the intersection can only be between the components $(V_i^\perp\times\{e_i\})$ and $(V_j^\perp\times\{e_j\})+\gamma$. For each $v\in V_i^\perp$ such that $v+\gamma_x\in V_j^\perp$, the vector $(v,e_i)$ belongs to the intersection, since
$$
(v,e_i)+\gamma = (v+\gamma_x,\, e_j) \in V_j^\perp\times\{e_j\}\subset \mathcal S.
$$
Similarly, the vector $(v+\gamma_x,e_j)$ also belongs to the intersection, because
$$
(v+\gamma_x,e_j)+\gamma = (v,e_i)\in \mathcal S.
$$
Thus each pair $(v, v+\gamma_x)$ with $v\in V_i^\perp\cap(V_j^\perp+\gamma_x)$ corresponds to exactly two elements of the intersection. By \cite[Claim 3.6]{HHL+} (applied to $V_i,V_j$ with $V_i\cap V_j=\{0\}$), the number of such $v$ is $2^{m-2d}$. Hence
$$
|\mathcal S\cap(\mathcal S+\gamma)| = 2\cdot 2^{m-2d}.
$$
In our case $m=2d+1$, so $|\mathcal S\cap(\mathcal S+\gamma)| = 2\cdot 2 = 4$, which does not exceed $2^{d+1}$ for all $d\ge1$. Note that $|\mathcal O_\gamma| = |\mathcal S\cap(\mathcal S+\gamma)|/2 = 2$, which is consistent with Titsworth's theorem allowing exactly two pairs in a folding.

\item If $|\gamma_y|=0$, then $\gamma=(\gamma_x,0)$. Then
$$
|\mathcal S\cap(\mathcal S+\gamma)| = \sum_{i=1}^N |V_i^\perp \cap (V_i^\perp+\gamma_x)| = \sum_{i:\gamma_x\in V_i^\perp} |V_i^\perp|.
$$
By Lemma~\ref{spread_lemma_3}, the number of indices $i$ with $\gamma_x\in V_i^\perp$ is at most $2$ (for any nonzero $\gamma_x$; the case $\gamma_x=0$ gives $\gamma=0$, which is excluded). Therefore,
$$
|\mathcal S\cap(\mathcal S+\gamma)| \le 2\cdot 2^{d+1} = O(2^{d+1}) = O(|\mathcal S|^{1/2}).
$$
\end{itemize}

Thus for any nonzero $\gamma$,
$$
|\mathcal S\cap(\mathcal S+\gamma)| \le 2^{d+2}= O(|\mathcal S|^{1/2}).
$$

\begin{remark}
The estimate $|\mathcal S\cap(\mathcal S+\gamma)| \le 2^{d+2}$ is tight: there exists a nonzero vector $\gamma=(\gamma_x,0)\in\mathbb F_2^{m+N}$ such that $|\mathcal S\cap(\mathcal S+\gamma)| = 2^{d+2}$.

Indeed, the subspaces $V_i^\perp$ have dimension $d+1$ in $\mathbb F_2^m$, where $m=2d+1$, and their number is $N=2^{d+1}$. If every nonzero vector belonged to at most one $V_i^\perp$, then the union of all $V_i^\perp$ would contain
$$
1 + N(2^{d+1}-1) = 1 + 2^{d+1}(2^{d+1}-1) = 2^{2d+2} - 2^{d+1} +1
$$
vectors, which is strictly greater than $2^{2d+1}$ (the size of the whole space) for all $d\ge1$. Hence there exists a nonzero vector $\gamma_x$ belonging to at least two distinct subspaces $V_i^\perp$. By Lemma~\ref{spread_lemma_3}, it cannot belong to more than two, so it belongs to exactly two. For such $\gamma_x$, with $\gamma=(\gamma_x,0)$, we have
$$
|\mathcal S\cap(\mathcal S+\gamma)| = \sum_{i:\gamma_x\in V_i^\perp} |V_i^\perp| = 2\cdot 2^{d+1} = 2^{d+2}.
$$
\end{remark}

Since $|\mathcal S\cap(\mathcal S+\gamma)| = |(\mathcal S+\gamma_1)\cap(\mathcal S+\gamma_2)|$ for $\gamma=\gamma_1+\gamma_2$, we obtain the following theorem.

\begin{theorem}\label{thm:improved}
For every $d\ge1$ there exists a Boolean function $f$ on $n = (2d+1)+2^{d+1}$ variables with spectrum support $\mathcal S$ satisfying
$$
|\mathcal S| = 2^{2d+2},
$$
and for any distinct $\gamma_1,\gamma_2$,
$$
|(\mathcal S+\gamma_1)\cap(\mathcal S+\gamma_2)| = O(|\mathcal S|^{1/2}).
$$
\end{theorem}

This substantially strengthens the folding estimate in the construction underlying Theorem 1.3 of \cite{HHL+} (or, more precisely, Theorem 3.1 of \cite{HHL+}).
We construct a different family of Boolean functions for which the exponent $5/6$ in the folding estimate in \cite{HHL+}  is replaced by $1/2$. Note that $n$ in our construction is of order $2^{d+1}$, and $|\mathcal S|$ is exponential in $d$, so we obtain an infinite series of examples.

In the next section, we discuss the consequences of this theorem for greedy algorithms for constructing parity decision trees.

\section{Consequences for greedy approaches to PDT depth}

In this section we discuss what lower bounds on the depth of parity decision trees (PDTs) follow from restrictions on the size of the maximum intersection of spectrum supports, and how our improved estimate of Theorem~\ref{thm:improved} affects the prospects of greedy methods.

We recall the basic construction. Let $f:\mathbb F_2^n\to\{-1,1\}$ be a Boolean function, $\mathcal S=\operatorname{supp}\widehat f$. For an arbitrary nonzero $\gamma\in\mathbb F_2^n$, consider the two subfunctions corresponding to the restrictions to the hyperplanes $\langle \gamma,x\rangle=0$ and $=1$.
Instead of the restricted functions themselves, it is convenient first to consider their extensions by zero to the whole space:
$$
f_b(x)=f(x)\frac{1+(-1)^{\langle\gamma,x\rangle+b}}2,\qquad b\in\{0,1\}.
$$
Their Fourier coefficients satisfy
$$
\widehat{f_b}(\alpha)=
\frac12\left(\widehat f(\alpha)+(-1)^b\widehat f(\alpha+\gamma)\right).
$$
For the actual restriction to the hyperplane $\langle\gamma,x\rangle=b$, the vectors
$\alpha$ and $\alpha+\gamma$ are identified.

Let
$$
s=|\mathcal S|,\qquad t=|\mathcal S\cap(\mathcal S+\gamma)|.
$$

The number $t$ counts elements of $\mathcal S$ whose $\gamma$-partner is also in $\mathcal S$. Therefore the number of folded pairs is $t/2$.

After restriction in the direction $\gamma$:
\begin{itemize}
\item
vectors from $\mathcal S\setminus(\mathcal S+\gamma)$ appear in both branches;
\item
a folded pair $\{\alpha,\alpha+\gamma\}\subset\mathcal S$ contributes to at least one of the two branches, possibly to both, depending on whether cancellation occurs.
\end{itemize}
Thus for the actual restricted functions $g_0,g_1$,
$$
[ |\operatorname{supp}\widehat{g_0}|+ |\operatorname{supp}\widehat{g_1}| \ge 2(s-t)+\frac t2 =
2s-\frac{3t}{2}=2|\mathcal S| - |\mathcal S\cap(\mathcal S+\gamma)|.
$$

Consequently, at least one branch satisfies

$$
|\operatorname{supp}\widehat{g_b}| \ge s-\frac{3t}{4}=
|\mathcal S|-\frac{3}{4}\cdot |\mathcal S\cap(\mathcal S+\gamma)|.
$$

If $t=|\mathcal S\cap(\mathcal S+\gamma)|\le C|\mathcal S|^\beta$ with $\beta<1$, then after the query
there exists at least one branch whose Fourier support has size at least
$$
|\mathcal S|-O(|\mathcal S|^\beta).
$$
Thus, for lower-bounding the depth of a decision tree, one may follow such a branch.

Now consider a greedy algorithm for constructing a PDT.

Suppose that during a PDT construction, at every current node with Fourier support
\(\mathcal S_i\), one has
\begin{equation}\label{eq:tag_3_20}
\max_{\gamma\ne0}|\mathcal S_i\cap(\mathcal S_i+\gamma)|
\le C|\mathcal S_i|^\beta
\end{equation}
for some \(\beta<1\). Then there exists a branch along which
\[
|\mathcal S_{i+1}|\ge |\mathcal S_i|-C'|\mathcal S_i|^\beta.
\]
Consequently, the tree has depth at least
\[
\Omega(|\mathcal S_0|^{1-\beta}).
\]
In particular, if \(\beta=1/2\), this argument gives a lower bound
\[
\Omega(|\mathcal S_0|^{1/2})
\]
on the depth of such a query process. Since the general upper bound
\[
\operatorname{PDT}(f)=O(|\mathcal S|^{1/2})
\]
is already known, an argument relying only on such maximum-folding estimates cannot improve the general state of the art.

Our Theorem~\ref{thm:improved} provides a construction with $\beta=1/2$. So, assuming the inheritance of estimate \eqref{eq:tag_3_20} for all restrictions, we obtain the following.

\begin{corollary}
Suppose that along every branch produced by a query strategy the current Fourier support
$\mathcal S_i$ satisfies
$$
\max_{\gamma\ne0}|\mathcal S_i\cap(\mathcal S_i+\gamma)|
\le
C|\mathcal S_i|^{1/2}.
$$
Then this strategy has a branch of length at least
$$
\Omega(|\mathcal S|^{1/2}).
$$
In particular, an analysis based only on such maximum-folding estimates cannot yield a
general PDT upper bound better than $O(|\mathcal S|^{1/2})$.
\end{corollary}

This result is final in the following sense: it is known that for every Boolean function $f$, $\mathrm{PDT}(f)=O(|\mathcal S|^{1/2})$ (see \cite{TWXZ13,MS24}). Thus the lower bound $\Omega(|\mathcal S|^{1/2})$ matches the upper bound up to constants, meaning that the greedy approach, within this analysis, cannot yield anything better than the already known general upper bound. In other words, to obtain polylogarithmic depth (which is required by the log-rank conjecture for XOR functions), one must use more subtle properties beyond just tracking the maximum folding.

However, it should be emphasised that this does not imply an actual PDT lower bound for our functions, because the condition above need not be inherited by restrictions.  In general, this assumption is false. For example, for the classical address function $\operatorname{Add}_m$ (and its generalisations), the PDT depth is $O(m)=O(\log|\mathcal S|)$, even if its spectrum support does not have relatively large folding. This means that for some order of decompositions along directions, the folding size inevitably increases significantly when passing to restrictions. Consequently, the claim that the greedy approach is refuted holds only for the \emph{lazy} variant, where the researcher hopes to obtain a lower bound on PDT depth by simply substituting a global bound like \eqref{eq:tag_3_20} but {\it lower} (i.e., with an inverted inequality sign) at each step without analysing the specific structure of the restrictions.

It remains an open question whether one can construct a greedy algorithm that adaptively analyses the current support at each step and finds directions with large intersections, even though the global estimate \eqref{eq:tag_3_20} may fail.

Our construction shows that no universal lower bound of the form
$$
\max_{\gamma\ne0}|\mathcal S\cap(\mathcal S+\gamma)|
\ge c|\mathcal S|^\beta
$$
can hold for any $\beta>1/2$.
At the same time it do not rule out the possibility of using other ideas to achieve polylogarithmic depth. Thus, although the <<lazy>> greedy approach is refuted, a complete refutation of greedy strategies in general remains open.

\end{document}